\documentclass[12pt]{article}
\usepackage{amssymb}
\usepackage{amsmath}
\usepackage{graphicx}
\usepackage[matrix,arrow]{xy}
\parskip=\medskipamount
\usepackage{xcolor}

\newenvironment{proof}{\par\smallskip\noindent\textbf{Proof:}}{\par\medskip}

\newtheorem{theorem}{Theorem}

\newtheorem{proposition}[theorem]{Proposition}

\def\IC{\relax{\it  l\kern-.50 em C}}
\def\IE{\relax{\it  l\kern-.12 em E}}
\def\IK{\relax{\it  l\kern-.18 em K}}
\def\IL{\relax{\it  I\kern-.18 em L}}
\def\IN{\relax{\it  I\kern-.18 em N}}
\def\IR{\relax{\it  I\kern-.18 em R}}

\def\<#1>{\langle#1\rangle}
\def\d<#1>{\langle\!\langle#1\rangle\!\rangle}
\def\D<#1>{\left\langle\!\!\left\langle#1\right\rangle\!\!\right\rangle}

\font\tenfrak=eufm10  \font\sevenfrak=eufm7  \font\fivefrak=eufm5
\newfam\frakfam
\textfont\frakfam=\tenfrak
\scriptfont\frakfam=\sevenfrak
\scriptscriptfont\frakfam=\fivefrak

\newtheorem{example}{Example}

\def\dfrac#1#2{{\displaystyle{\frac{#1}{#2}}}}

\def\pd#1#2{\frac{\partial #1}{\partial#2}}

\begin{document}

\title{Conformal Killing vector fields and a  virial theorem}
\author{
Jos\'e F. Cari\~nena$^{a)}$, Irina Gheorghiu$^{a)}$, Eduardo\ Mart\'{\i}nez$^{b)}$ and Patr\'{\i}cia Santos$^{c)}$
\\[4pt]
$^a$Department of Theoretical Physics, Univ. of Zaragoza, Spain\\
$^b$IUMA and Department of Appl. Math., Univ. of Zaragoza, Spain\\ 
$^c$CMUC, Univ. of Coimbra, and  Polytech. Inst.  of Coimbra, ISEC, Portugal 
}
\date{}
\maketitle


\begin{abstract}
The virial theorem is formulated both intrinsically and in local coordinates for a Lagrangian system of mechanical type on a Riemann manifold. An import case studied in this paper is  that of an affine virial function associated to a vector field on the configuration manifold. The special cases of a virial function associated to a Killing, a homothetic and a conformal Killing vector field are considered and the corresponding virial theorems are established for this type of functions.
\end{abstract}
\begin{quote}

{\sl Running title:}{\enskip}
Conformal vector fields and a  virial theorem.




\end{quote}
\section{Introduction}  Since the establishment of the so called virial theorem in 1870 its usefulness and range of applicability have been increasing almost continuously 
till our days. It was  stated by Clausius in the assertion {\it The mean vis viva of the system is equal to its virial}
where {\it vis viva} integral is the total kinetic energy of the system
and the latin word {\it virias}  was used by Clausius to denote
the scalar quantity represented in terms of the forces ${\bf F}_i$ acting on the system as  
$$\frac 12 \d<\sum_{i}{\bf F}_i\cdot {\bf r}_i>,
$$
and it was  shown to be one half of the averaged potential energy of the
 system. 

The important point is the  wide range of applicability of the virial theorem, 
as  it is applicable to dynamical and thermodynamical systems,  it can also be formulated to deal with relativistic (in the sense of special relativity) systems, 
it is applicable to systems with velocity dependent forces and viscous systems, and even if  it provides less information that the equations of motion,  it is simpler to apply and then it can provide information concerning systems whose complete analysis may defy description. 
For instance, in astronomy, the virial theorem finds applications in the theory of dust and gas of interstellar space as well as cosmological considerations of the universe as a whole and in other  discussions concerning the stability of clusters, galaxies and clusters of galaxies. For an excellent   historical account one can see \cite{GC78}.

 In one-particle Newtonian mechanics of a  particle of  mass $m$ under the action of a force ${\bf F}$ the virial  function introduced by Clausius is
 $G({\bf x}, \dot {\bf x})=m\, {\bf x}\cdot \dot{\bf x}$,
 and one can show  using Newton second law that ${dG}/{dt}=m\,\dot {\bf x}\cdot \dot{\bf x}+{\bf x}\cdot {\bf F}$, and when integrating this expression 
 between $t=0$ and $t=\tau$, dividing by the total time interval $\tau$ and taking the limit of $\tau $ going to infinity we find that  if the possible values of $G$ 
 are bounded then  $\d<2\, T( \dot {\bf x})+ {\bf x}\cdot {\bf F}>=0$. In the particular case of a conservative force, ${\bf F}=-\boldsymbol{\nabla}V$, 
 $\d<2\, T( \dot {\bf x})- {\bf x}\cdot \boldsymbol{\nabla}V>=0$. When 
  the potential  $V$ is homogeneous of degree $k$, Euler's theorem of homogeneous functions implies that ${\bf x}\cdot \boldsymbol{\nabla}V=k\, V$, and therefore,
 $\d<2\, T( \dot {\bf x})-k\,V({\bf x})>=0$, 
 i.e. if $E$ is the total energy,  
 $$
\d<T( \dot {\bf x})>=\frac{k\,E}{k+2},\qquad \d<V({\bf x})>=\frac{2\,E}{k+2}.$$

Remark that the existence of the time average of a function depends on the evolution curve and therefore on the initial conditions. The assumption that the function remains bounded guarantees that such an  
average does exist whatever the evolution curve be, and, moreover, when the motion is periodic the average coincides with the average in a time period. On the other side the relation $\d<A+B>=\d<A>+\d<B>$
holds when the three averages do exist. In particular, expressions as $\d<A-B>=0$ imply $\d<A>=\d<B>$ when $\d<A>$ does exist.

Relevant questions about this result are: Where does the virial function $G$ comes from?
 Why the relation is simpler for power law potentials? Why is the reason for the values of the coefficients? 
 Is there any generalisation? and, of course,  What about a quantum mechanical counterpart? The answer to all these questions rests on the geometrical 
 interpretation of the virial theorem for dynamical systems, in particular for systems  defined by regular Lagrangians,  (see e.g. \cite{CFR12} and references therein). 
  The standard virial theorem  is based on the transformation properties of kinetic and potential
energy under dilations, therefore is only valid for systems with $\mathbb{R}^n$ as configuration space.
In order to generalise the virial theorem for other system we should use the tools of geometric mechanics.  
First the problem was analysed in the framework of Hamiltonian dynamical systems, and therefore for systems described by a regular Lagrangian.
But  the possibility of establishing a virial-like relation in the case in which we have a vector field $X$ which is a complete lift of a vector field on the base manifold such that $XL=a\,L$ was also studied in \cite{CFR13}, 
as well as the even  more general case 
of vector fields whose flows are non-strictly canonical transformations.   The virial-like relations we obtain are more general than the standard ones and some of them have been used in tensor
 virial theorems or the so-called hypervirial theorems \cite{Hi79}.

It is well known the equivalence of Lagrangian and Hamiltonian formalisms in the regular case. Actually the geometric approach was first developed in the Hamiltonian
 formalism in the framework of 
symplectic geometry in phase spaces (i.e. cotangent bundles) and then Legendre transformation was used to translate the symplectic structure to the Lagrangian formalism. However,
 it was soon proved that one can develop the formalism in the framework of tangent bundle geometry by using the geometric tensors characterizing tangent bundle structures, 
 the vertical endomorphism and the Liouville vector field \cite{Go,Kl,Cr81,Cr83}. As virial-like relations can be directly established in terms of the Lagrangian function and are
  not so easily derivable in the Hamiltonian formalism we will mainly restrict ourselves to the Lagrangian formalism, even if the final expressions can be translated to
  the Hamiltonian language. The extension of some of such results to the framework of mechanics in Lie algebroids was developed in 
\cite{CGMS14}.

This paper tries to develop analogous results in the particular case of mechanical type Lagrangians, and in this case conformal Killing vector fields 
will be shown to play a very relevant role. For mechanical systems, $L=T_g-V$, finding infinitesimal symmetries of the metric, i.e. Killing fields, is relatively easy. 
As it is well known, if such a vector field is also a symmetry of the potential we get a constant of the motion, which simplifies the problem. 
If the Killing vector field is not a symmetry of the potential 
  the Virial Theorem provides relevant information, namely the average value of the derivative of the potential vanishes. With more generality, for a homothetic or a conformal Killing vector field the Virial theorem allows us to establish relations between the averages of the kinetic energy and those of certain
 derivatives of the potential. 
  
The paper is organized as follows. In Section 2, some geometrical concepts about Riemann structures are recalled and relevant expressions for tensor fields and functions are written in generalised coordinates. In Section 3, the virial theorem is presented for Lagrangian systems of mechanical type, both in intrinsic form and in terms of local coordinates, and a spherical geometry problem is analysed using this approach. In Section 4, we consider an important particular case of an affine on the velocities virial function, associated to a vector field on the configuration  manifold, more specifically, when the vector field is either a Killing, a homothetic or a conformal Killing vector field. Several examples are used to illustrate the theory.
In the last section we make some final comments about the results presented in the paper.

\section{Riemann structures and mechanical type Lagrangians}
Let $(M,g)$ be a (pseudo-)Riemann manifold, i.e. $g$ is a non-degenerate  symmetric two times covariant tensor field on $M$. Nondegeneracy means that the map $\widehat g:TM\to T^*M$ from the tangent bundle $\tau_M:TM\to  M$ to the cotangent bundle $\pi_M:T^*M\to M$,  defined by $\<\widehat g(v),w>=g(v,w)$, where  $v,w\in T_xM$, is regular. The map $\widehat g$ is a fibred map over the identity on $M$ and induces the corresponding map between the spaces of sections of the tangent and cotangent bundles, to be denoted by the same letter $\widehat g:\mathfrak{X}(M)\to \Omega^1(M)$: $\<\widehat g(X),Y>=g(X,Y)$.

A diffeomorphism $F:M\to M$ induces a new (pseudo-)Riemann structure $F^*g$ on $M$. Such transformation $F$ is called a conformal symmetry when there exists a function $f\in C^\infty(M)$ such that $F^*g=f\,g$. In particular when $f$ is a constant (different from one) $F$ is said to be a (proper) homothethy and, finally, when $F^*g=g$, the map  $F$ is called isometry. In the infinitesimal approach we say that a vector field  $X\in \mathfrak{X}(M)$ is either a conformal,   a homothetic, or a Killing vector field,
 when its flow $\phi_t$ is made of 
conformal maps, homothethies or isometries, respectively:
$$
\begin{array}{ll}{\rm conformal\ vector\ field:}&\mathcal{L}_Xg=f\, g,\quad  f\in C^\infty(M),\\
{\rm homothetic\ vector\ field:}&\mathcal{L}_Xg=\lambda\, g,\qquad \lambda\in \mathbb{R},\\
{\rm Killing\ vector\ field:}&\mathcal{L}_Xg=0. \end{array}
$$
Proper conformal vector fields are those vector fields for  which the conformal factor  $f$ is non constant and similarly a proper  
homothetic\ vector\ field is when $\lambda\ne 0$. Using the well known property $\mathcal{L}_X\circ \mathcal{L}_Y-\mathcal{L}_Y\circ \mathcal{L}_X=\mathcal{L}_{[X,Y]}$ one sees that the set of conformal vector fields is a Lie algebra and those of homothetic and Killing vector fields are subalgebras. For more details see e.g. \cite{GH88,GH00,TPBC,MMM}.

Given a symmetric covariant 2-tensor field $K$ in $M$ we denote by $T_K\in C^\infty(TM)$ the function 
\[
T_K(v)=\frac{1}{2}K(v,v),\qquad v\in TM.
\]
This rule identifies symmetric covariant  2-tensor fields with quadratic homogeneous  functions on the fibre coordinates. In particular when $g$ is a Riemann structure in $M$, 
$$
T_g(v)=\frac{1}{2}g(v,v),\qquad v\in TM,
$$
is the kinetic energy defined by the metric.

Given a local chart $(U,q^1,\ldots,q^n)$ on $M$ we can  consider the coordinate basis of $\mathfrak{X}(U)$ usually denoted $\{\partial/\partial q^j\mid j=1,\ldots ,n\}$ and the dual basis for $\Omega^1(U)$, $\{dq^j\mid j=1,\ldots ,n\}$. Then a  vector in a point $q\in U$  is $v=v^j\,(\partial/\partial q^j)_q$ and a covector  is $\zeta=p_j\,(dq^j)_q$, with  $v^j=\<dq^j,v>$ and $p_j=\<\zeta,\partial/\partial q^j>$ being the usual velocities and momenta. The local expression for 
$g$ is 
\begin{equation}
g=g_{ij}(q)\,dq^i\otimes dq^j.\label{Rstr}
\end{equation}

We can define Lagrangians of mechanical type for systems with configuration space  $M$, $L \in C^\infty(TM)$, by choosing a (pseudo-)Riemann structure  $g$ on $M$ and a potential function $V\in C^\infty(M)$ as follows:
\begin{equation}
L_{g,V}(q,v)=\frac 12 \,g_q(v,v)-(\tau_M^*V)(q,v)=\frac 12 \,g_q(v,v)-V(q),\label{mechL}
\end{equation}
i.e. the Lagrangian function is of the form $L_{g,V}=T_g-\tau_M^* V$, where the function $T_g\in C^\infty (TM)$  represents the kinetic  energy given above which can be rewritten as   
$$
T_g=\frac{1}{2}\, g(T\tau_M\circ D,T\tau_M\circ D),
$$
with  $D$ being  any second order differential equation vector field, i.e. a vector field on $TM$ such that $\tau_{TM}\circ D={\rm id}_{TM}$, 
while  the potential energy $\widetilde V=\tau_M^*V$ is a  basic function, i.e. the pull-back of  a smooth function $V$ on the base manifold $M$.

Given  a  Riemann structure $g$ on a manifold $M$ with   local expression   in a local chart (\ref{Rstr}),
the expression for the corresponding (free, i.e. $V=0$) Lagrangian, i.e. the function $T_g$,  is 
\begin{equation}
T_g(q,v)=\frac 12\, g_{ij}(q)\,v^iv^j ,\label{Lfree}
\end{equation}
while the coordinate expression of an arbitrary  second order vector field is 
\begin{equation}
D(q,v) =v^i\, \pd{}{q^i}+f^i(q,v)\pd{}{v^i}\, .\label{sodefield}
\end{equation}

Given  a vector field on $M$,
\begin{equation}
X=X^i(q)\, \pd{}{q^i}\in\mathfrak{X}(M),\label{vfinM}
\end{equation}
the Lie derivative with respect to the vector field $X$ of the metric tensor field $g$ is 
$$\mathcal{L}_Xg=X^k\,\pd{g_{ij}}{q^k} \, dq^i\otimes dq^j+g_{ij}\left(\pd {X^i}{q^k}\, dq^k\otimes dq^j+\pd{X^j}{q^k}dq^i\otimes dq^k\right),
$$
or using the symmetry property of the metric tensor field,
\begin{equation}
\mathcal{L}_Xg=\left(X^k\,\pd{g_{ij}}{q^k} +g_{ik}\,\pd{X^k}{q^j}+g_{jk}\,\pd{X^k}{q^i}\right) dq^i\otimes dq^j,\label{LXg}
\end{equation}
and then the condition for $X$ 
to be a Killing vector field, i.e. $\mathcal{L}_Xg=0$, is written in the above mentioned local coordinates as 
$$
\left(X^k\,\pd{g_{ij}}{q^k} +g_{ik}\,\pd{X^k}{q^j}+g_{jk}\,\pd{X^k}{q^i}\right)\, dq^i\otimes dq^j=0.
$$

Therefore, the set of conditions for the vector field $X\in\mathfrak{X}(M)$ given by (\ref{vfinM}) to be a Killing symmetry are:
\begin{equation}
X^k\,\pd{g_{ij}}{q^k} +g_{ik}\,\pd{X^k}{q^j}+g_{jk}\,\pd{X^k}{q^i}=0, \qquad i,j=1,\ldots, n.\label{Killingcond}
\end{equation}

Consider now the complete lift $X^c\in \mathfrak{X}(TM)$ with flow $T\phi_t$, where $\phi_t$ is the flow of the vector field $X\in\mathfrak{X}(M)$ with local expression (\ref{vfinM}).   Then the  local  coordinate expression of $X^c$  is 
$$
X^c(q,v)=X^i(q)\, \pd{}{q^i}+v^j\,\pd{X^i}{q^j}(q,v)\,\pd{}{v^i}=X^i(q)\, \pd{}{q^i}+(DX^i)(q,v)\pd{}{v^i}, 
$$
for any second order differential equation vector field $D$. 

For a 1-form $\alpha$ on $M$ we denote by $\widehat{\alpha}$ the associated linear function on $TM$ given by $\widehat{\alpha}(v)=\langle\alpha_{\tau_M(v)},v\rangle$, for $v\in TM$. In local tangent bundle coordinates, if $\alpha=\alpha_i(q)\, dq^i$, the function $\widehat{\alpha}$ is $\widehat{\alpha}(q,v)=\alpha_i(q)\, v^i$. In particular, for an exact 1-form $\alpha=df$ the associated linear function is $\widehat{df}(q,v)=v^i(\partial f/\partial q^i)_q$, i.e.  $\widehat{df}$ looks like the total derivative of the function $f$, and we denote $\dot{f}=\widehat{df}$, which can also be obtained by $\dot{f}=\mathcal{L}_{D}(\tau_M^*f)$ for an arbitrary second order differential equation vector field $D$. Complete lifts are determined by the action on this kind of functions: given a vector field $X$ on $M$ its complete lift $X^c$ is the only vector field on $TM$ which satisfies 
\begin{equation}
\mathcal{L}_{X^c}\widehat{\alpha}=\widehat{\mathcal{L}_X\alpha}\label{LXca}
\end{equation}
for every 1-form $\alpha$ on $M$. It is clear that the above  condition determines a vector field on $TM$. Let us show that the complete lift satisfies such a relation. If $\phi_t$ is the flow of $X$ then the flow of $X^c$ is $T\phi_t$, so that for $v\in TM$, with $q=\tau_M(v)$, we have
\begin{align*}
(\mathcal{L}_{X^c}\widehat{\alpha})(v)
&=\frac{d}{dt}\widehat{\alpha}(T\phi_t(v))\big|_{t=0}=\frac{d}{dt}\langle\alpha_{\tau_M(T\phi_t(v))}, T\phi_t(v)\rangle \big|_{t=0}\\
&=\frac{d}{dt}\langle\alpha_{\phi_t(q)}, T\phi_t(v)\rangle \big|_{t=0}=\frac{d}{dt}\langle(\phi_t^*\alpha)_q, v\rangle \big|_{t=0}\\
&=\langle(\mathcal{L}_X\alpha)_q, v\rangle=\widehat{\mathcal{L}_X\alpha}(v).
\end{align*}
In particular, for $\alpha=df$ we have $\mathcal{L}_{X^c}\dot{f}=(\mathcal{L}_Xf)\,\dot{}$. 

A remarkable property to be used later on is that for a  given a vector field  $X\in\mathfrak{X}(M)$,  $[X^c,D]$ is  a vertical vector field in $TM$ for any  second order differential equation vector field $D$, because $X^cD(q^i)=D(X^i)=DX^c(q^i)=v^k\partial X^i/\partial q^k$. The preceding  property (\ref{LXca}) can also  be used to give an intrinsic proof  as follows. 
Indeed, the action on basic functions is
\begin{align*}
\mathcal{L}_{[D,X^c]}(\tau_M^*f)
&=\mathcal{L}_{D}\mathcal{L}_{X^c}(\tau_M^*f)-\mathcal{L}_{X^c}\mathcal{L}_{D}(\tau_M^*f)\\
&=\mathcal{L}_{D}(\tau_M^*\mathcal{L}_{X}f)-\mathcal{L}_{X^c}\dot{f}=(\mathcal{L}_{X}f)\dot{}\,-\mathcal{L}_{X^c}\,\dot{f}=0
\end{align*}
from where it follows that $[X^c,D]$ is vertical.

One of the more important properties of complete lifts is the following relationship:
\begin{equation}
X^cT_g=T_{\mathcal{L}_Xg}.\label{XcTg}
\end{equation}
In fact, 
\begin{align*}
(X^cT_g)(q,v)
&=\frac 12 \left(X^k(q) \pd{g_{ij}}{q^k}(q)\, v^iv^j+g_{ij}(q)\, \pd{X^i}{q^k}(q)\,v^kv^j+g_{ij}(q)\,\pd{X^j}{q^k}(q)\,v^iv^k\right)\\
&=\frac 12 \left(X^k(q)\pd{g_{ij}}{q^k}(q)+g_{kj}(q)\, \pd{X^k}{q^i}(q)+g_{ik}(q)\, \pd{X^k}{q^j}(q)\right)v^iv^j
\end{align*}
and therefore, according to (\ref{LXg}),  the relation (\ref{XcTg}) follows. This relation may also be proved intrinsically by using the definitions of Lie derivative and of $T_g$ mentioned earlier in the text: for all $v\in TM$,
$$\begin{array}{rcl}
X^cT_g(v)&=&{\displaystyle \frac{d}{dt} T_g\circ T\phi_t(v)\big|_{t=0}
=\frac{d}{dt}\left( \frac{1}{2}g(T\phi_t(v),T\phi_t(v))\right)_{|t=0}}\\
&=&{\displaystyle\frac 12   \frac{d}{dt} (\phi_t^*g)(v,v)\big|_{t=0}
=\frac{1}{2} (\mathcal{L}_X g)(v,v)=T_{\mathcal{L}_Xg}(v,v)}.
\end{array}
$$

Consequently, $X\in \mathfrak{X}(M)$ is a Killing vector field for the Riemann structure  $g$ if and only if $X^c\in \mathfrak{X}(TM)$ is a symmetry for the corresponding free Lagrangian, i.e. 
 the conditions for $X^c$ to be a symmetry of $T_g$ are given by  (\ref{Killingcond}).
 
\section{A virial theorem for mechanical type Lagrangians}

A (regular)  Lagrangian determines a symplectic structure on the tangent bundle $TM$, the Cartan 2-form $\omega_L=-d\theta_L=-d(dL\circ S)$. Here $S$ is the vertical endomorphism \cite{Cr81, Cr83},
which is defined using the natural identification of the tangent space $T_qM$ with the vertical subspace of the tangent space in any point of $\tau_M^{-1}(q)$.   Such a vertical lift allows
 us to lift a tangent vector field $X\in \mathfrak{X}(M)$ to a vertical vector field $X^{\mathrm{v}}\in \mathfrak{X}(TM)$. This vector field is related to the complete lift by $S(X^c)=X^{\mathrm{v}}$,
  and if the local coordinate expression of $X$ is (\ref{vfinM}) that of $X^{\mathrm{v}}$ is 
  $$
  X^{\mathrm{v}}(v)=X^i(q)\, \pd{}{v^i}\in \mathfrak{X}(TM) .
  $$

  The energy of a Lagrangian system is defined by $E_L=\Delta L - L$, where $\Delta$ is the Liouville vector field, generator of dilations along the fibres, 
  given by $$\Delta f(q,v)=\frac{d}{dt} f(q,e^tv)|_{t=0},$$ for all $(q,v)\in TM$ and $f\in C^\infty(M)$. Hence, as $\Delta(T_g)=2\, T_g$ and $\Delta(V)=0$, 
  the total energy of a Lagrangian system of mechanical type  is $E_L=T_g+V$. 
  
The dynamics is then given by the dynamical vector field $\Gamma_L$ defined for a regular Lagrangian $L$ by 
\begin{equation}
i(\Gamma_L)\omega_L=dE_L.\label{dyneq}
\end{equation}
 In particular,  the coordinate expression of  the Cartan 1-form  $\theta_L=dL\circ S$ 
for the Lagrangian (\ref{mechL}) is given by
$$\theta_L(q,v)=g_{ij}(q)\,v^j\, dq^i$$
and the symplectic form $\omega_L=-d\theta_L$ by 
$$
\omega_L= g_{ij}\, dq^i\wedge dv^j +\frac{1}{2}\left(\frac{\partial g_{ij}}{\partial q^k}v^j-\frac{\partial g_{kj}}{\partial q^i}v^j\right) dq^i\wedge dq^k.
$$

A regular Lagrangian system on $M$ can be seen as a Hamiltonian system $(TM,\{\cdot,\cdot\},H)$, where the Hamiltonian $H$ is the energy $E_L$ and the Poisson bracket $\{\cdot,\cdot\}$ is defined by the symplectic 
 2-form $\omega_L$, i.e. $\{F_1,F_2\}=\omega_L(X_{F_1},X_{F_2})$, where $i(X_F)\omega_L=dF$. Recall that a Poisson bracket  is a skew-symmetric bilinear map on the algebra of smooth functions 
 on the manifold, that obeys the Jacobi identity and the Leibniz's rule w.r.t. the first argument. In this case, the Virial Theorem states (see e.g. \cite{CFR12} and references therein)  that for a smooth bounded function $G$ the time average of the Poisson bracket $\dot G=\{G,E_L\}$ vanish, that is,  $\d<X_G(E_L)>=0$.
 
We next recall some important geometric properties of 
 connections. Recall that a linear  connection  on a Riemann  manifold $(M,g)$ is compatible with  the Riemann structure $g$, i.e. the parallel transport along
  any curve is an isometry,   if and only if 
 \begin{equation}
 X(g(Y,Z))=g(\nabla_XY,Z)+g(Y,\nabla_XZ),\qquad \forall X,Y,Z\in\mathfrak{X}(M).\label{Rcon}
 \end{equation}
 The main result is that there exists a unique torsion-free metric connection on $M$, called  Levi-Civita connection, which is given by Koszul formula:
\begin{equation}
 \begin{array}{rcl}
 2g(\nabla_XY,Z)&=&Xg(Y,Z)+Yg(Z,X)-Zg(X,Y)\\&
-&g(X,[Y,Z])+g(Y,[Z,X])+g(Z,[X,Y ]).
\end{array}
\end{equation}
In particular, when a coordinate chart is considered, the Christoffel symbols of the second kind defined by 
$$
\nabla_{\partial/\partial q^j} \left(\pd{}{q^k}\right)=\Gamma _{jk}^i\pd{}{q^i}
$$
 are given by
 \begin{equation}
 \Gamma_{jk}^i(q)=\frac 12 g^{il}(q)\left( \pd{g_{lj}}{q^k}(q)+ \pd{g_{lk}}{q^j}(q)- \pd{g_{jk}}{q^l}(q)\right),\label{Gijk}
\end{equation}
where $g^{ij}$ are the inverse matrix entries of the Riemann structure $g$.
 
 Then, the linear connection is given by 
 $$
 \nabla_XY=X^i\left(\pd{Y^k}{q^i}+Y^j\ 
 \Gamma_{ij}^k(q)\right)\pd{}{q^k},$$
 and correspondingly, 
 $$
  \nabla_X\alpha= X^k\left(\pd{\alpha_j}{q^k}-\alpha_i\,\Gamma^i_{jk}\right) dq^j.
  $$
 
 Using these covariant derivatives the Killing  condition $\mathcal{L}_Xg=0$, i.e. (\ref{Killingcond}),  can be written in an intrinsic way as the condition for the covariant derivative 
 of the vector field $X$, to be a skew-symmetric endomorphism with respect to the metric $g$, that is  (see e.g. Proposition 4.10 of \cite{CvW}),
  for every $Y,Z\in\mathfrak{X}(M)$,
  \begin{equation}
  g(\nabla_YX,Z)+g(Y,\nabla_ZX)=0.\label{Killingcond2}
    \end{equation}

 Another remarkable relation is that if $\alpha$ is the 1-form $\alpha=\widehat g(X)$, where $X	\in\mathfrak{X}(M)$, then,
 using that the relation 
$\nabla_Z\<\alpha,Y>=\<\nabla_Z\alpha,Y>+\<\alpha,\nabla_ZY>$, for any two vector fields  $Y,Z	\in\mathfrak{X}(M)$, 
can be rewritten then as
$$Z(g(X,Y))=\<\nabla_Z\alpha,Y>+g(X,\nabla_ZY),
$$
and having in mind  the property of the compatibility of the connection with the metric,  we see  that 
\begin{equation}
\<\nabla_Z\alpha,Y>=g(\nabla_ZX,Y). \label{duality}
\end{equation}

The dynamical vector field, solution of the dynamical equation (\ref{dyneq})  turns out to be
$$
\Gamma_L(q,v)= v^i\pd{}{q^i}- \left( \Gamma_{jk}^i(q)v^jv^k+ g^{ij}(q)\frac{\partial V}{\partial q^j}(q)\right)\pd{}{v^i},
$$
where  $\Gamma_{jk}^i$ are the Christoffel symbols of the second kind with respect to 
 the Levi-Civita connection defined by the metric $g$, as given by (\ref{Gijk}).
 
 The Hamiltonian vector field of a smooth function $G$ on $TM$ is determined by the equation $i(X_G)\omega_L= dG$ and in local coordinates  is given  by
\begin{equation}
X_G(q,v)= g^{ij}(q)\frac{\partial G}{\partial{v^j}}(q,v) \pd{}{q^i}
+g^{ik}(q)\left[ \left(\frac{\partial g_{ln}}{\partial q^k}(q)\,v^n-\frac{\partial g_{kn}}{\partial q^l}(q)\,v^n\right) g^{lj}(q)\frac{\partial G}{\partial{v^j}}(q,v)
- \frac{\partial G}{\partial{q^k}}(q,v)\right ] \pd{}{v^i}.\label{XG}
\end{equation}

Since the total energy of the system is $E_L=T+V$, then,
\begin{equation}
X_G(E_L)= -\Gamma_L(G)=\frac{\partial G}{\partial{v^l}}\left( \Gamma_{jk}^lv^jv^k+ g^{il}\frac{\partial V}{\partial q^i}\right)-\frac{\partial G}{\partial{q^k}} v^k.
\end{equation}
The   virial theorem, $\d<X_G(E_L)>=0$ (see e.g. \cite{CFR12} for a geometric approach), establishes
 the following relation between time averages:
$$
\D<\frac{\partial G}{\partial{v^l}}\left( \Gamma_{jk}^lv^jv^k+ g^{il}\frac{\partial V}{\partial q^i}\right)-\frac{\partial G}{\partial{q^k}} v^k>=0.
$$
We will see that the preceding expression is much simpler when the vector field $X_G$ is a complete lift. 

A relevant result concerning the virial theorem is  that  
if  $X$ is a vector field on $M$ and $X^c$ its complete lift, then the function $G$ defined by $G=\langle\theta_L,X^c\rangle$ is such that  $\mathcal{L}_{\Gamma_L}G=\mathcal{L}_{X^c}L$, that is, 
\begin{equation}
\Gamma_L(G)=X^c(L). \label{GLGXcL}
\end{equation}
In fact, as $L$ is assumed to be regular the vector field $\Gamma_L$ satisfies 
$\mathcal{L}_{\Gamma_L}\theta_L=dL$ and then $\langle \mathcal{L}_{\Gamma_L}\theta_L-dL,X^c\rangle=0$. Using a well-known property of the Lie derivative, 
$$
\langle \mathcal{L}_{\Gamma_L}\theta_L,X^c\rangle=i(X^c)\mathcal{L}_{\Gamma_L}\theta_L=\mathcal{L}_{\Gamma_L}i(X^c)\theta_L+i([X^c,\Gamma_L])\theta_L,
$$we have
\begin{eqnarray*}
\Gamma_L(\langle\theta_L,X^c\rangle)-\langle\theta_L,[\Gamma_L,X^c]\rangle-\langle dL,X^c\rangle&=&0.
\end{eqnarray*}
But the Cartan 1-form $\theta_L$ is a semi-basic 1-form and  $[X^c, \Gamma_L]$ is a vertical vector field because $\Gamma_L$ is a second order vector field and then 
$\langle\theta_L,[\Gamma_L,X^c]\rangle=0$. Therefore, $\Gamma_L(\langle\theta_L,X^c\rangle)=\Gamma_L(G)=\langle dL,X^c\rangle=X^c(L)$.

From the expression $\Gamma_L(G)=X^c(L)$, evaluating on the time evolution and averaging on the interval $[0,\tau]$, in the limit when $\tau\to\infty$, we get 
as we did in \cite{CFR12} in an analogous case,  that if $G$ remains bounded,
$$
\d<X^c(L)>=0 \Longleftrightarrow \d<X^c(T_g)-X(V)>=0,
$$
whose local coordinate expression is 
\begin{equation}
\D<X^k\frac{1}{2}\, \pd{g_{ij}}{q^k}v^iv^j+\pd{X^k}{q^l}g_{kj}v^lv^j-X^k \pd{V}{q^k}>=0.
\end{equation}

In the particular case studied in \cite{CFR12}, in which there exists a nonzero real number $a$ such that $X^cL=a\,L$ we recover the result $\d<L>=0$, i.e. $\d<T-V>=0$.
\begin{example}[Spherical geometry]\label{SGeo}
Consider as an illustrative example the motion of a unity mass point on a sphere of radius $R=1/\sqrt{\lambda}$ centred at the origin and  the usual spherical polar coordinates, i.e. a point $P$ on the sphere is fixed by two coordinates $(\theta,\phi)$ such that 
$$
\mathbf{x}(\theta,\phi)=(R\,\sin\theta\,\cos\phi, R\,\sin\theta\,\sin\phi, R\,\cos\theta), 
$$
and then 
$$
g_{\theta\theta}=R^2,\quad  g_{\theta\phi}=0,\quad g_{\phi\phi}=R^2\,\sin^2\theta,
$$
i.e. the arc-length is 
\begin{equation}
ds^2=R^2(d\theta^2+\sin^2\theta\,d\phi^2).\label{gsg}
\end{equation}
Suppose that the motion is under the action described by a potential function $V(\theta)$
that does not depend on $\phi$ but only on the distance to the North pole.
Then,  if $X$ is the vector field on the base $X=\tan\theta\, \partial/\partial \theta$, with complete lift 
$$X^c=\tan \theta\pd{}{\theta}+\sec ^2\theta\,  v_\theta\pd{}{v_\theta},
$$
as the kinetic energy is $T=\frac 12R^2(v_\theta^2+\sin^2\theta\,v_\phi^2) $ and 
$$
X^c(T)=R^2(\sec ^2\theta\,  v_\theta^2+ \sin^2\theta\,v_\phi^2), \qquad X(V)=\tan\theta\,\pd V\theta,
$$
the Virial Theorem establishes that
$$
\D<{R^2(\sec ^2\theta\,  v_\theta^2+ \sin^2\theta\,v_\phi^2)}>=\D<\tan\theta\,\pd V\theta>.
$$

The points of the lower half sphere can be described by the points obtained by central projection onto the tangent plane $x_3=-R$, i.e. points $(q_1,q_2,-R)$ such that
$$\left\{\begin{array}{rcl}q_1&=&\dfrac {x_1\, R}{-x_ 3}=-\dfrac{R^2\,\sin\theta\,\cos\phi}{R\,\cos\theta}=-R\tan\theta\,\cos\phi \\ q_2&=&\dfrac {x_2\, R}{-x_ 3}=-\dfrac{R^2\,\sin\theta\,\sin\phi}{R\,\cos\theta}=-R\tan\theta\,\sin\phi\end{array}\right.
$$
or eliminating the South pole and using polar  coordinates $(r,\phi)$ centred at $(0,0,-R)$, i.e. $r=-R\tan\theta$, having in mind that 
$$\frac{d\theta}{dr}=-\frac 1R\ \frac 1{1+(r/R)^2}=-\frac 1R\ \frac 1{1+\lambda\,r^2},
$$
the expression of  the arc-length becomes
$$ds^2=\frac 1{(1+\lambda r^2)^2} dr^2+\frac{r^2}{(1+\lambda r^2)}d\phi^2.
$$

 In terms of the new coordinates,
as $\tan\theta=-r/R$,
$$
\sec^2\theta=1+\lambda r^2,\qquad \sin^2\theta=\frac {r^2}{R^2} (1+\lambda r^2)^{-1},\qquad v_r=R(1+\lambda r^2)v_\theta
$$
and then we can rewrite the preceding equation as 
\begin{equation}
\D<(1+\lambda r^2)^{-1}(v_r^2+r^2v_\phi^2)>=\D<r\,(1+\lambda r^2)\pd Vr>,\label{liform}
\end{equation}
which coincides with the expression (14) of \cite{LZC11}. However, in~\cite{LZC11}  such expression was only proved for two special cases and it was proposed as a guess for the general case.
\end{example}

\section{Affine virial functions}

As mentioned earlier, the Virial Theorem  for a given smooth bounded function $G$ is but $\d<X_G(E_L)>=0$, which for systems of mechanical type reduces to $\d<X_G(T)+X_G(V)>=0$. A particularly simple case would be when $X_G$ is a complete lift and this property constraints the possible form of $G$. 
 
 A particularly simple case would be when $X_G$ is a complete lift and this 
 property constraints the possible form of $G$.  
  
Note first that the expression (\ref{XG}) for the vector field $X_G$ shows that  the necessary and sufficient condition for $X_G$ to be $\tau_M$ projectable is that $\partial G/\partial v^i$ be a basic function, i.e. $G$ is an affine in velocities function, or in more geometric language,  there must be a 1-form $\alpha=\alpha_k(q)\, dq^k$ on $M$ and a function $\varphi$  on $M$ such that 
$$G=\widehat \alpha+\tau_M^*\varphi,$$ 
and then the $\tau_M$-related vector field is $\widehat g^{-1}(\alpha)$.

\subsection{Killing vector fields}

In order to the vector field $X_G$ to be  a complete lift,  the $\tau_M$-related vector field  must be $\widehat g^{-1}(\alpha) $,  and the  $n$ functions $\alpha_k$  and the function $\varphi$ on the base manifold must satisfy, for any index $i$, 
$$
\pd{}{q^k}\left(g^{ij}\alpha_j\right)v^k=g^{ik}\left[ \left(\frac{\partial g_{ln}}{\partial q^k}v^n-\frac{\partial g_{kn}}{\partial q^l}v^n\right) g^{lj}\alpha_j
- v^j\frac{\partial \alpha_j}{\partial{q^k}}-\pd{\varphi}{q^k}\right ] .
$$
These conditions can be rewritten for any pair of indices $(i,k)$,   as:
$$
\alpha_j\,\pd{g^{ij}}{q^k}+g^{ij}\,\pd{\alpha_j}{q^k}=g^{ij}\pd{g_{lk}}{q^j}\,g^{lm}\,\alpha_m-g^{im}\,\pd{g_{mk}}{q^l}\,   g^{lj}\,\alpha_j
-\pd{\alpha_k}{q^j}\,g^{ij},\qquad  \pd{\varphi}{q^k}=0,
$$
and therefore as follows
$$
g^{ij}\left(\pd{\alpha_j}{q^k}+\pd{\alpha_k}{q^j}\right)=\alpha_n\left(-\pd{g^{in}}{q^k}+g^{ij}g^{ln}\pd{g_{lk}}{q^j}-g^{im}g^{ln}\pd{g_{mk}}{q^l}\right) , \qquad  \pd{\varphi}{q^k}=0.
$$

Using now that 
$$
\pd{g^{ij}}{q^k}= -g^{il}g^{jm}\pd{g_{lm}}{q^k},
$$
the preceding equation becomes
$$
g^{ij}\left(\pd{\alpha_j}{q^k}+\pd{\alpha_k}{q^j}\right)=\alpha_n\left(g^{ir}g^{ns}\pd{g_{rs}}{q^k}+g^{ij}g^{ln}\pd{g_{lk}}{q^j}-g^{im}g^{ln}\pd{g_{mk}}{q^l}\right),
$$
or equivalently
$$g^{ij}\left(\pd{\alpha_j}{q^k}+\pd{\alpha_k}{q^j}\right)=\alpha_n g^{ij}g^{ln}\left(\pd{g_{jl}}{q^k}+\pd{g_{lk}}{q^j}-\pd{g_{jk}}{q^l}
\right)=2 \alpha_n g^{ln}\Gamma^i_{lk}.
$$
which can be rewritten as
$$
\pd{\alpha_j}{q^k}+\pd{\alpha_k}{q^j}=2\,\alpha_i\,\Gamma^i_{jk},
$$
or in other words,  for any pair of indices $i$, $k$,
\[
\left(\pd{\alpha_j}{q^k}-\alpha_i\,\Gamma^i_{jk}\right)+\left(\pd{\alpha_k}{q^j}-\alpha_i\,\Gamma^i_{kj}\right)=0.
\]
Multiplying both sides by $Z^jY^k$ and summing on repeated indices we see that 
 this equation is the coordinate expression of the intrinsic one  
 $$\<\nabla_Y\alpha,Z>+\<\nabla_Z\alpha,Y>=0,\qquad \forall Y,Z\in\mathfrak{X}(M),$$
 so that the 2-covariant tensor field $\nabla\alpha$ is skew-symmetric. But as $\alpha=\widehat g(X)$, the relation (\ref{duality}) allows us to express this condition as 
 $g(\nabla_YX,Z)+g(Z,\nabla_YX)=0$, which means that $X$  
satisfies  the Killing condition (\ref{Killingcond2}). The preceding result can be summarized  in the following proposition  whose intrinsic proof is also given:

\begin{proposition} The vector field $X\in\mathfrak{X}(M)$ is a Killing vector w.r.t. 
the Riemann structure $g$ iff $X_{\widehat\alpha}=X^c$, where $\widehat\alpha$
is the linear in the fibres function defined by the 1-form $\alpha=\widehat g(X)$.
\end{proposition}

\begin{proof} The linear in the fibres function  $G=\<\theta_{T_g}, X^c>$ is nothing but the function $\widehat\alpha$, because
$$
\<\theta_{T_g},X^c>= \<dT_g\circ S,X^c>= \<dT_g, S(X^c)>=X^{\mathrm{v}}(T_g),
$$
where the vector field  $X^{\mathrm{v}}$ is the vertical lift of $X$ \cite{Cr81,Cr83}, and therefore,
\[
\<\theta_{T_g},X^c>(v)=\frac{d}{ds}T_g(v+sX(\tau_M(v))\big|_{s=0}=g(X(\tau_M(v)),v)=\widehat \alpha(v),
\]
for every $v\in TM$.

 If the Hamiltonian vector field  $X_G$ is the complete lift $X^c$, then the relation (\ref{GLGXcL}) shows that $X^c(E_L)=-X^c(L)$, because $X^c(L)=
\Gamma_LG=-X_G(E_L)=-X^c(E_L)$. Therefore, $X^c(T_g)-X^c(V)=-X^c(T_g)-X^c(V)$, i.e.
 $X^c(T_g)=0$, and then $X$ is a Killing vector. On the other hand, if $X$  is a Killing vector we have that $T_{\mathcal{L}_Xg}=0$. Since $i(X_G-X^c)\omega_{T_g}=\theta_{T_{\mathcal{L}_Xg}}=0$, then $X_G=X^c$.\end{proof}

Let $X$ be a Killing vector field, and $\alpha=  \widehat{g}(X)$ the associated 1-form. As we have seen, $X_{  \widehat{\alpha}}=X^c$, from where we have
\[
\{E_L,  \widehat{\alpha}\}=X_{  \widehat{\alpha}}E_L=X^cE_L=E_{X^cL}=T_{\mathcal{L}_Xg}+\tau_M^*\left({\mathcal{L}_XV}\right)=\tau_M^*\left({\mathcal{L}_XV}\right).
\]
Taking mean values we get that for every Killing vector field $X$:
\[
\d<\mathcal{L}_XV>=0.
\]
Therefore, if $X$ is not a symmetry of the potential energy then the mean value of the derivative $\mathcal{L}_XV$ vanishes along any trajectory of the Lagrangian dynamical system.

\begin{example}[Spherical geometry revisited] 

Coming back to the case of the spherical geometry, we can say that the vector field 
$$X=X_\theta\,\pd{}\theta+X_\phi\, \pd{}\phi$$
is a Killing vector field if and only if its complete lift
$$X^c=X_\theta\,\pd{}\theta+X_\phi\, \pd{}{\phi}+\left(\pd{X_\theta}\theta\,v_\theta+\pd{X_\theta}\phi\, v_\phi\right)\pd{}{v_\theta}+\left(\pd{X_\phi}\theta\,v_\theta+\pd{X_\phi}\phi\,v_\phi\right)\, \pd{}{v_\phi}
$$
is a symmetry of the kinetic energy
$$T(\theta,\phi,v_\theta,v_\phi)=\frac 12 (v_\theta^2+\sin^2\theta\, v_\phi^2).
$$
From the condition 
$$\left(\pd{X_\theta}\theta\,v_\theta+\pd{X_\theta}\phi\, v_\phi\right)v_\theta+\sin^2\theta\left(\pd{X_\phi}\theta\,v_\theta+\pd{X_\phi}\phi\,v_\phi\right)v_\phi+X_\theta\,\sin\theta\,\cos\theta\,{v_\phi}^2=0,
$$
we obtain the conditions:
$$\begin{array}{rl} &{\displaystyle\pd{X_\theta}\theta}=0,\\&\\&{\displaystyle\pd{X_\theta}\phi+\sin^2\theta\pd{X_\phi}\theta}=0\\
&\sin\theta\left(\cos\theta\,X_\theta+\sin\theta{\displaystyle\pd{X_\phi}\phi}\right)=0
\end{array}
$$

One solution is given by $X_\theta=0$ and $X_\phi=1$, i.e. the vector field $X_3=\partial/\partial \phi$ is a Killing vector field. Another particular solution is $X_\theta=\cos\phi$ and $X_\phi=-\sin\phi\,\mathrm{cotan\,}\theta$,
and then another  Killing  vector field is
$$X_1=\cos\phi\pd{}\theta-\sin\phi\,\mathrm{cotan\,}\theta\pd{}\phi.
$$
The corresponding virial theorem is
\[
\d<\mathcal{L}_{X_1}V>=0 \Longleftrightarrow 
\d<\cos\phi\pd{V}\theta>=\d<\sin\phi\,\mathrm{cotan\,}\theta\pd{V}\phi>.
\]
\end{example}

\begin{example}[Periodic Toda lattice with $n$ particles]
A periodic Toda lattice system with $n$ particles without impurities (each particle as the same mass $m$), is defined by a mechanical Lagrangian $L=T-V$ on $T\mathbb{R}^n$. The kinetic energy is the quadratic function defined by the Euclidian metric on $\mathbb{R}^n$,
$$
T(q,v)=\frac{1}{2}\sum_{i=1}^n m\,v_i^2,
$$
and the potential is given by
$$
V(q)=\sum_{i=1}^{n} e^{q_i-q_{i+1}},
$$
where $q_{n+1}=q_1$. Consider the following vector field, for a fixed $k=1,\ldots,n$,
$$
X_k=\frac{\partial}{\partial q_{k}}.
$$
The vector field is a Killing vector w.r.t. the Euclidean metric.

Then the Virial Theorem implies that $\d<\mathcal{L}_{X_k}V>=\d<e^{q_k-q_{k+1}}-e^{q_{k-1}-q_{k}}>=0$. Therefore, $\d<e^{q_k-q_{k+1}}>=\d<e^{q_{k-1}-q_{k}}>$ for every $k$ and hence $\d<V>=n\d<e^{q_1-q_{2}}>$.

\end{example}

\begin{example}[Kepler problem in polar coordinates]
Consider a particle $P$ of mass $m$ moving in a plane under the action of a central force $F(r)=-\gamma\, m\,m'/r^2$ on the direction of a fixed point $O$ of mass $m'\gg m$, where $\gamma$ is a positive constant and $r$ represents the distance between $O$ and the point particle $P$. Let $\phi$ be the angle that the line $OP$ makes with a fixed direction on the plane. In polar coordinates the arc-length is given by $ds^2=dr^2+r^2d\phi^2$. The kinetic energy of the particle is given by
$$
T(r,\phi,v_r, v_\phi)=\frac{m}{2}\left(v_r^2+r^2\,v_{\phi}^2\right)
$$
and the potential is the function $V(r)=-\gamma\, m\, m'/r$. The vector field
$$
X=\cos\phi\,\pd{}{ r}-\frac 1r\sin\phi\pd{}{\phi}
$$
is a Killing vector field of the Euclidean metric in polar coordinates.
Then the Virial Theorem tell us that $\d<\mathcal{L}_XV>=0$, that is, $\d<-\cos(\phi)\gamma\, m\, m'/r^2>=0$.

\end{example}

\subsection{Conformal Killing and homothetic vector fields}

Conformal Killing vector fields and in particular homothetic vector fields have also been relevant in many problems in physics and more particularly in space-time geometry (see e.g, \cite{MMM,CT71,GM07}). We now explore the information that we can extract from them in the problem of virial theorem we are considering. With this aim  we first find the difference between the Hamiltonian vector field $X_{\widehat \alpha}$ associated to the 1-form  $\alpha=  \widehat{g}(X)$,  where $X$ is  a vector  field on $M$ and the complete lift of $X$.

\begin{proposition}
If  $X$ is the vector field on $M$  associated to the 1-form $\alpha$, $\alpha=  \widehat{g}(X)$, and as before $  \widehat{\alpha}\in C^\infty(TM)$ is the function $  \widehat{\alpha}(v)=g(X(\tau_M(v)),v)$, for $v\in TM$,  then the difference of the complete lift $X^c$ of $X$ and the Hamiltonian vector field $X_{  \widehat{\alpha}}$ associated to $  \widehat{\alpha}$ with respect to the symplectic form $\omega_{T_g}$ is the  vertical vector field  whose contraction with the symplectic form $\omega_{T_g}$ is the semi-basic 1-form 
$ \theta_{T_{\mathcal{L}_Xg}}$.
\end{proposition}
\begin{proof}
Notice first that  as both vector fields,  $X^c$  and $X_{\widehat{\alpha}}$, are projectable on the vector field $X=\widehat g^{-1}(\alpha)$, the difference vector is vertical. Moreover, taking  into account the above mentioned relation $\<\theta_{T_g},X^c>=  \widehat{\alpha}$, 
we have
\begin{equation}
i(X_{  \widehat{\alpha}}-X^c)\omega_{T_g}
=i(X_{  \widehat{\alpha}})\omega_{T_g}-i(X^c)\omega_{T_g}=d  \widehat{\alpha}+i(X^c)d\theta_{T_g},
\end{equation}
and then 
\begin{equation}i(X_{  \widehat{\alpha}}-X^c)\omega_{T_g}
=d(i(X^c)\theta_{T_g})+i(X^c)d\theta_{T_g}
=\mathcal{L}_{X^c}\theta_{T_g}=\theta_{X^cT_g}
=\theta_{T_{\mathcal{L}_{X}g}},\label{XaXc}
\end{equation}
where the last equality follows from (\ref{XcTg}).
\end{proof}

It is also well known (see e.g. \cite{CR93})  that contraction with the symplectic forms $\omega_L$ defined by a regular Lagrangian $L$ establishes a one-to-one correspondence of vertical vector fields with semi basic 1-forms. More explicitly, in the particular case we are considering of $L=T_g$, the semi basic 1-form corresponding to the Liouville vector field $\Delta$, generating dilation along the fibres of $TM$,  is $-\theta_{T_g}$ because, as $\theta_{T_g}$ is semi-basic, 
$$i(\Delta)\omega_{T_g}= -i(\Delta)d\theta_{T_g}= -\mathcal{L}_{\Delta}\theta_{T_g},
$$
and as $\theta_{T_g}$ is homogeneous of degree one in velocities, we find that 
\begin{equation}
i(\Delta)\omega_{T_g}=-\theta_{T_g}.\label{DwTg}
\end{equation}
This allows us to write:
$$i(X_{  \widehat{\alpha}}-X^c)\omega_{T_g}=-i(\Delta)\omega_{T_{\mathcal{L}_Xg}}.
$$

As a consequence, in the case of a conformal Killing vector field, we have the following result.
\begin{theorem}
A vector field $X$ on $M$  is a conformal Killing vector field, i.e.  there exists a function $f\in C^\infty(M)$ such that $\mathcal{L}_Xg=f\,g$, if and only if 
$X_{  \widehat{\alpha}}=X^c-f\,\Delta$, where $\alpha $ is the 1-form $\alpha=\widehat g(X)$.
\end{theorem}
\begin{proof}
Indeed, if $X$ is a conformal Killing vector field,   there exists a function $f\in C^\infty(M)$ such that 
$\mathcal{L}_{X}g=fg$, and then $\theta_{T_{\mathcal{L}_Xg}}=f\,\theta_{T_g}$. The relation (\ref{XaXc}) reduces in this case to  $i(X_{  \widehat{\alpha}}-X^c)\omega_{T_g}
=f\,\theta_{T_g}$, and then using (\ref{DwTg}), to $i(X_{  \widehat{\alpha}}-X^c)\omega_{T_g}=-i(f\,\Delta)\omega_{T_g}$. As $\omega_{T_g}$ is nondegenerate 
we find $X_{  \widehat{\alpha}}-X^c=-f\,\Delta$.

Conversely, if  there exists a function $f\in C^\infty(M)$ such that $X_{  \widehat{\alpha}}-X^c=-f\Delta$, then 
$$i(X_{  \widehat{\alpha}}-X^c)\omega_{T_g}=-i(f\,\Delta)\omega_{T_g}=f\,\theta_{T_g},
$$
and as a consequence of (\ref{XaXc}) we obtain that $\theta_{T_{\mathcal{L}_{X}g}}=f\,\theta_{T_g}$, which implies $\mathcal{L}_Xg=f\,g$ and then $X$ is a conformal Killing vector field.
\end{proof}

This result is in agreement with the meaning  of being a conformal Killing vector field: its flow transforms geodesics in re-parameterized geodesics, the responsible for reparametrization is the term $f\,\Delta$. 
Of course, for $f=0$ we recover the result of Proposition 1.  

The preceding results allow us to introduce a Virial Theorem for conformal Killing vector fields.
\begin{theorem}
Let us consider a Lagrangian of mechanical type $L=L_{g,V}=T_g-\tau_M^*{V}$, a conformal Killing vector  field $X$ for $g$, i.e.  there exists a function $f\in C^\infty(M)$ such that $\mathcal{L}_Xg=f\,g$, and the associated 1-form
$\alpha=\widehat g(X)$. Then 
we have that 
\[
\d<fT_g-\mathcal{L}_XV>=0.
\]
\end{theorem}
\begin{proof}  If  $\alpha=\widehat g(X)$ is the associated 1-form,
from  the relation $X_{  \widehat{\alpha}}=X^c-f\Delta$ it follows that 
\[
\{E_L,  \widehat{\alpha}\}=X_{  \widehat{\alpha}}E_L=X^cE_L-f\,\Delta E_L=E_{X^cL}-2fT_g=-fT_g+\tau_M^*({\mathcal{L}_XV}),
\]
where we have used that $E_{X^cL}=T_{\mathcal{L}_Xg}+\tau_M^*(\mathcal{L}_XV)=f\,T_g+\tau_M^*(\mathcal{L}_XV)$. 
Applying the virial theorem $\d<\{E_L,  \widehat{\alpha}\}>=0$ we obtain the result.
\end{proof}

\begin{example} 

Consider now the spherical geometry metric (\ref{gsg}) for $R=1$ and look for a conformal vector field of the form $X=X_\theta(\theta) \, \partial/\partial \theta$. From the relationships 
$$\mathcal{L}_X(d\theta^2)=2 \dot X_\theta\, d\theta^2,\qquad \mathcal{L}_X(\sin^2\theta\,d\phi^2)= 2\,\sin\theta\,\cos\theta\,X_\theta\,d\phi^2$$
we see that in order to be a conformal vector field one must have:
$$2\,\dot X_\theta=2\,\mathrm{cotan\,} \theta\,X_\theta=f(\theta),
$$
from where we obtain $X_\theta=\sin\theta $ and $f(\theta)=2\,\cos\theta$. Therefore the corresponding Virial relation  reads 
$$
\d<2\,\cos\theta\,T_g>=\d<\,\cos\theta\,(v_\theta^2+\sin^2\theta\,v_\phi^2) >=\d<\sin\theta \pd V\theta>.
$$
\end{example}

\begin{example} 

Another example with three degrees of freedom is the metric
$$ds^2=h(r)\,dr^2+r^2(d\theta^2+\sin^2\theta\, d\phi^2),         \qquad h(r)>0.
$$
If we  look for a conformal vector field of the form $X=X_r(r) \, \partial/\partial r$ we arrive to the relationship
$$\mathcal{L}_X\,g
=( \dot hX_r+2\,h\,\dot X_r)\, dr^2+2\,r\,X_r(d\theta^2+\sin^2\theta\, d\phi^2)=f\,g,
$$
and we see that in order to be a conformal vector field one must have:
$$\frac{ \dot h}hX_r+2\,\dot X_r=\frac{2X_r}r=f,
$$
from where we can conclude that  $X_r$ is a solution of the differential equation 
$$\dot X_r+\left(\frac 12 \,\frac{ \dot h}h-\frac 1r\right)X_r=0\Longrightarrow X_r=C\frac r{h^{1/2}}.
$$ and $f=2C/h^{1/2}$. In particular for $h(r)=1$, the Euclidean metric, we have the homothetic dilation vector field $X=r \, \partial/\partial r$, with $f=2$  while
 for $h(r)=r^2$ we  find the conformal vector field $X= \partial/\partial r$ with a conformal factor $f=2/r$.Therefore the corresponding Virial relations  read
$$
\d<2T_g>=\d<\mathcal{L}_XV>\Longrightarrow \d<v_r^2+r^2(v_\theta^2+\sin^2\theta\, v_\phi^2)>=\D<r \pd Vr>.
$$
and 
$$
\d<\frac 2rT_g>=\d<\mathcal{L}_XV>\Longrightarrow \d<r\left(v_r^2+v_\theta^2+\sin^2\theta\, v_\phi^2\right)>=\D<\pd Vr>.
$$
\end{example}

We can prove a similar result when we have two Riemann metrics $g$ and $g'$ on $M$ and the vector field $X\in\mathfrak{X}(M)$ relates them in the following way $\mathcal{L}_Xg=f\,g'$.

\begin{theorem} Consider a Lagrangian of mechanical type $L=L_{g,V}=T_g-\tau_M^*{V}$. If there exists a function $f\in C^\infty(M)$ such that $\mathcal{L}_Xg=f\,g'$, and $\alpha=\widehat g(X)$, then,
\[
\d<fT_{g'}-\mathcal{L}_XV>=0.
\]
\end{theorem} 
\begin{proof} Since $\widehat\alpha=\<\theta_L,X^c>$,  then $\Gamma_L(\widehat\alpha)=X^c(L)$. Hence,
$$\{E_L, \widehat{\alpha}\}
=X_{ \widehat{\alpha}}E_L
=-X^c(L)
=-X^c(T_g)+\tau_M^*(\mathcal{L}_XV)
=-T_{\mathcal{L}_Xg}+\tau_M^*(\mathcal{L}_XV)
=-fT_{g'}+\tau_M^*(\mathcal{L}_XV).$$
The Virial Theorem implies that $\d<-fT_{g'}+\mathcal{L}_XV>=0$, and the result follows.
\end{proof}

\begin{example}[Spherical geometry]\label{ExConf}
In example \ref{SGeo} the vector field $X=\tan(\theta)\partial_\theta$ defines the virial function. In polar coordinates, this vector is given by
$$
X=r(1+\lambda r^2)\partial_r.
$$
The vector field $X$ is not a conformal vector field of the Euclidian metric $g'$ given by $ds^2=dr^2+r^2d\phi^2$, but $\mathcal{L}_X g=2(1+\lambda r^2)^{-1}g'$. In this case, we have $\d<2(1+\lambda r^2)^{-1}T_{g'}>=\d<\mathcal{L}_XV>$ and this formula is equivalent to (\ref{liform}).
\end{example}

A particularly interesting case is when the vector field  $X$ is homothetic, i.e.  $f=\mu$ is a real constant,
$\mathcal{L}_Xg=\mu\, g$, because then $\mathcal{L}_{X^c}T_g=\mu\, T_g$, where $T_g$ is the kinetic energy $T$.

In example \ref{ExConf}, when $\lambda\to 0$, the limit vector field is the infinitesimal generator of dilations on $\mathbb{R}^2$ written in polar coordinates, and it is a $2-$homothetic vector field of the Euclidian metric, so in the limit the Virial Theorem implies that $2\d<T_g>=\d<r\partial_r V>$. 

If $V$ is a $X$-homogeneous function of degree $\nu$, i.e. $\mathcal{L}_X V= \nu V$, then $\d<\mu T_g-\nu V>=0$ because of $\d<X^c(T_g)-\mathcal{L}_X V>=0$. Using that the energy is a constant $E$ along a trajectory we also have $\d<T_g+V>=E$, from where
\[
\d<T_g>=\frac{\nu}{\nu+\mu}E
\qquad\text{and}\qquad
\d<V>=\frac{\mu}{\nu+\mu}E.
\]

As a particular case, if both degrees of homogeneity are equal $\nu=\mu\equiv a$ then we have that 
\[
\d<T_g>=\d<V>=\frac{1}{2}E.
\]
On the other hand, this condition is equivalent to $\mathcal{L}_{X^c}L=aL$, and hence we can apply directly a result in~\cite{CFR12} obtaining  $\d<L>=0$, from where we also get $E=2\d<T_g>=2\d<V>$.

\section{Summary and outlook}

This paper aims to  go deeper in a recently developed geometric approach to the Virial Theorem and our  attention has been focused on the particularly interesting case of 
Lagrangian systems of mechanical type. Geometric properties of the Riemann structure defining the kinetic energy term allow us to identify different types of 
virial functions and associated vector fields for which a virial like type theorem can be stated. Recall that a first generalisation of virial theorem  was 
established in the framework  of  the theory of
Hamiltonian systems on symplectic manifolds and therefore for systems defined by regular Lagrangians. This is here studied in the particular case of systems of mechanical type Lagrangians. The general case is $\d<X_G(E_L)>=0$. But there are other type of virial like theorems of a specifically Lagrangian nature. For instance we have proved in Section 3 that for any complete lift vector field $X^c$ the following relation is true: $\d<X^c(T_g)-\mathcal{L}_X(V)>=0$; this is a generalisation of the case studied in \cite{CFR12} in which $X^c(L)=a\, L$. We have displayed in Section 4 the most general form of a function $G$ whose associated Hamiltonian vector field is the complete lift of a vector field $X$ on the base manifold,
which is a Killing vector field, the function $G$ being determined by the 1-form on $M$ corresponding to  $X$ by contraction with the Riemann metric. We obtain in this way the virial relation $\d<\mathcal{L}_XV>=0$. Finally we have identified the conformal Killing vectors and obtained a virial relation for such vector fields, the case of homothetic vector fields appearing as a particular case: $\d<fT_g-\mathcal{L}_XV>=0$.
Several examples have been used to illustrate the theory.

The usefulness of this geometric approach suggests us the convenience of analyzing 
the virial theorem in the framework of nonholonomic system. They are receiving  an increasing interest from the geometric viewpoint \cite{LMdD} and the use of a modern approach to the concept of quasi-velocity \cite{CNCS07} using the geometric tools of Lie algebroids points out to consider virial-like relations in Lagrangian and Hamiltonian systems on Lie algebroids \cite{EM01} when nonholonomic constraints are present \cite{CLMM}. We believe this kind of applications deserves a more detailed study.

\section*{Acknowledgments}

 This work was supported by the research projects 
 MTM--2012---33575 (MINECO, Madrid) and DGA-E24/1 (DGA, Zaragoza). We also thank 
 the anonymous referees for their valuable comments and suggestions.

\end{document}